\documentclass[12pt]{article}
\usepackage{amsthm}
\usepackage[utf8x]{inputenc}
\usepackage{graphicx}
\usepackage{setspace}
\usepackage{fullpage}
\usepackage{multirow}
\usepackage{amssymb}
\usepackage{subfigure}
\usepackage{amsmath}
\usepackage{lineno}

\newtheorem{cor}{Corollary}
\newtheorem{defn}{Definition}

\newtheorem{theorem}{Theorem}
\newtheorem*{thm3}{Theorem 3}
\newtheorem*{thm2}{Theorem 2}

\title{Reproductive Value in Graph-structured Populations}
\author{Wes Maciejewski\\
Department of Mathematics\\
The University of British Columbia\\
Vancouver, British Columbia, Canada}

\begin{document}

\linenumbers

\maketitle

\begin{abstract}
Evolutionary graph theory has grown to be an area of intense study. Despite the amount of interest in the field, it seems to have grown separate from other subfields of population genetics and evolution. In the current work I introduce the concept of Fisher's (1930) reproductive value into the study of evolution on graphs. Reproductive value is a measure of the expected genetic contribution of an individual to a distant future generation. In a heterogeneous graph-structured population, differences in the number of connections among individuals translates into differences in the expected number of offspring, even if all individuals have the same fecundity. These differences are accounted for by reproductive value. The introduction of reproductive value permits the calculation of the fixation probability of a mutant in a neutral evolutionary process in any graph-structured population for either the moran birth-death or death-birth process.
\end{abstract}

\section{Introduction}

\doublespacing

Population structure has, for some time, been recognized as an important factor in determining the outcome of an evolutionary process. Structure can act to arrange individuals and produce evolutionary outcomes not seen in well-mixed populations \cite{nowakmay92}.  Early models considered an infinite number of islands of individuals, each linked by global dispersal \cite{wright31}. Subsequent work, like the stepping-stone model of \cite{kimuraweiss64, weisskimura65}, considered the spatial arrangement of these islands. These models were refined to the finite population case by considering a finite number of breeding deems linked by dispersal patterns \cite{levins69, levins70}.  Drawing on these earlier models, evolutionary graph theory has emerged as a convenient framework for modelling population structure \cite{liebermanhauertnowak05}.

An evolutionary graph $G$ is a collection of vertices $V$ and edges $E$ between them. The vertices are occupied by haploid individuals and the edges indicate who interacts with whom and where offspring disperse. Throughout this article I denote vertices by $v_i$ and the individual residing on $v_i$ by $i$. It is possible that the vertices are linked by two sets of edges, one indicating interactions and the other, replacements \cite{ohtsukinowakpacheco07}, but these two sets are often assumed to coincide, as they do in this article.

Since their introduction in \cite{liebermanhauertnowak05}, evolutionary graphs have become a well-studied representation of structured populations. The exact features of graphs that promote, or work against, cooperation are, however, still elusive. For highly symmetric (vertex-transitive) graphs exact results for any additive game undergoing a weak-selection evolutionary process have been obtained \cite{ohtsukihauertliebermannowak06, taylordaywild07}. This is the largest class of graphs for which results are known, encompassing many other results \cite{ohtsukinowak06, grafen07}. Actual interaction graphs are often highly non-symmetric \cite{santossantospacheco08} and it is of great interest to study evolution in these environments. 

Very few results have been obtained for non-symmetric graphs. There has been some interest in the role of vertex degree. Some work \cite{santossantospacheco08} has focused on the distribution of the degrees of vertices. Certain distributions (`scale-free') have been shown to promote altruistic and cooperative behaviours more than others (eg. regular graphs). These approaches have uncovered global features of graphs and a description of the process at the level of the individual is desirable. One of the challenges faced in the study of heterogeneous populations is dealing with individuals of differing quality. \emph{Reproductive value} \cite{fisher30} is a way of accounting for such differences. 

Antal, Redner, and Sood \cite{antalrednersood06} are perhaps the first to consider heterogeneous graphs at the individual scale. They have found that it is advantageous for the fitter mutant to occupy high-degree nodes in a Moran death-birth model (their `biased voter model') and lower-degree vertices in the birth-death process (their `biased invasion process'). This has been confirmed by subsequent research \cite{broomrychtarstadler11}. In the current article I show that these results, when phrased in terms of reproductive value \cite{fisher30, grafen06}, are two sides of the same coin. 

The work of \cite{antalrednersood06} and \cite{broomrychtarstadler11} focuses on the case of \emph{constant selection}, where the resident population has fecundity $1$ and a mutant with fecundity $r>1$ arises. The probability of this mutant taking over the entire population is calculated and compared against the neutral case of $r=1$. If this mutant fixation probability is greater, the mutant is advantageous. An extension of the results of \cite{antalrednersood06} and \cite{broomrychtarstadler11} to the case of a public-goods game, as in \cite{santossantospacheco08}, is highly desirable. I attempt to make headway by presenting an example that illustrates that a mutant individual can have greater evolutionary success depending on where it first emerges.

The main thrust of this article is a complete description of the fixation probability of an allele in any graph-structured population undergoing neutral drift. For a structured population of size $N$ with the property that all sites are equivalent---for example, degree-regular graphs---then this fixation probability is $1/N$, irrespective on which vertex the allele is first found. This is not the case for degree-heterogeneous graphs. In general, the fixation probability depends on the degree of the vertex on which the allele initially appears. In the current article I calculate these fixation probabilities for both the birth-death and death-birth Moran processes on any graph. A general rule is derived: fixation probability is positively associated with relative reproductive value. An allele will have a higher fixation probability if it first emerges on a vertex with a higher reproductive value in both the birth-death and death-birth processes. 

\section{Reproductive Value}

Reproductive value has been defined in various ways by different authors. The core of the definitions is the notion of long-term genetic share of a population. R.A. Fisher \cite{fisher30} first introduced reproductive value as a means of accounting for the differences in the reproductive output of different ages of females. Since that time reproductive value has been applied to age \cite{charlesworth80}, sex \cite{taylor90}, and spatially-structured \cite{rogerswillekens78} populations and has been placed on a rigorous mathematical footing \cite{grafen06}. At an intuitive level, the relative reproductive value of an individual $i$ is the probability that $i$ is the ancestor of a randomly chosen individual in a distant future generation \cite{taylorfrank96}. 

To define reproductive value, I suppose that the individuals in the population under consideration are neutral with respect to selection. That is, the genotype of an individual does not affect their fitness. Births and deaths occur at random in the population. Throughout this article I work with two Moran processes, which will be made explicit, that ensure a fixed population size. In the birth-death process, a birth occurs randomly in the population and the new offspring displaces a neighbouring individual, who dies. In the death-birth process an individual is chosen to die and a neighbouring individual is chosen at random to place an offspring on the newly vacated site. These birth and death probabilities are captured by a transition matrix $M$. Specifically, I define the $i$, $j$ entry of $M$ to be the probability $p_{ij}$ that the current individual $i$ is the offspring of individual $j$ produced during a birth/death event. This entry will differ depending on whether births preceed deaths or vice versa, and examples throughout the article will illustrate this. An individual may be unaffected by the birth/death event in which case we say that such an individual is ``from itself". 

\begin{figure}
\centering
\includegraphics[width=0.35\textwidth]{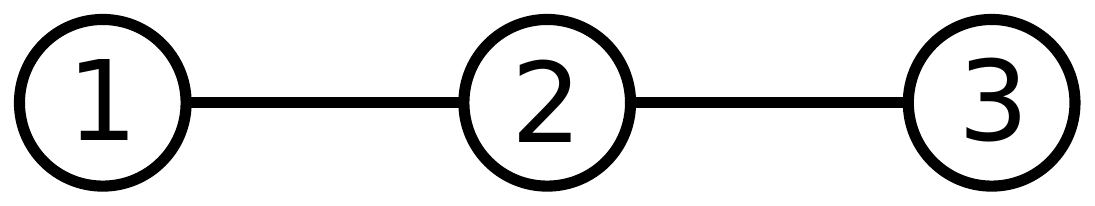}
\caption{The line graph on three vertices is the simplest example of a degree-heterogenous graph. Label the vertices $v_1$, $v_2$, and $v_3$ from left to right. As is shown in the penultimate section, spiteful behaviours can evolve in such a population structure and these depend on where the spiteful individual first emerges.}
\label{fig:3line}
\end{figure}

As a first example of such an $M$ matrix, consider a birth-death process on the $3$-line graph in Figure \ref{fig:3line}. In the neutral process all individuals have the same fecundity and are therefore chosen to reproduce with equal probability, which in the $3$-line case is $1/3$. If the centre, or hub, individual is chosen, then it places an offspring on either leaf vertex with probability $1/2$. If a leaf is chosen, its offspring disperses to the hub with probability $1$. Given the current state of the population, we can ask where the individual on a leaf vertex was before a birth-death event. With probability $1/6$, the individual is the offspring of the hub vertex and with probability $5/6$ the individual was unaffected by the birth-death event and was already resident on the leaf vertex. For the hub individual, with probability $1/3$ it came from one of leaf vertices and with probability $1/3$ it was unaffected by the birth-death event and already resident on the hub. In all, with the vertex numbering in Figure \ref{fig:3line},
\begin{eqnarray}
\label{eq:elinematrix}
M=
\left [ \begin{array}{c c c }
\frac{5}{6} & \frac{1}{6} & 0 \\
\frac{1}{3} & \frac{1}{3} & \frac{1}{3} \\
0 & \frac{1}{6} & \frac{5}{6} 
\end{array} \right].
\end{eqnarray}
This matrix $M$ can be used to find the vector of probabilities of the origin of the left-most leaf individual. Represent this individual with the vector $[1,0,0]$. This yields,
\begin{eqnarray}
\label{eq:3linematrix}
[1,0,0]\left [ \begin{array}{c c c }
\frac{5}{6} & \frac{1}{6} & 0 \\
\frac{1}{3} & \frac{1}{3} & \frac{1}{3} \\
0 & \frac{1}{6} & \frac{5}{6} 
\end{array} \right]=\left[\frac{5}{6}, \frac{1}{6} ,0\right],
\end{eqnarray}
which captures the argument above: with probability $5/6$ the leaf individual was unaffected by the birth-death event and with probability $1/6$ it is an offspring of the hub individual. Another right-multiplication by $M$ yields the probability vector for two generations previous. And so on. 

To find the probability that a randomly chosen individual in the population at a time $t$, measured in the number of birth/death events, in the future is from the lineage originating from individual $i$ at the present time $t_0 =0$, we perform a calculation similar to the above on the vector $[1,1,1]$:
\begin{eqnarray}
\label{eq:Mlimit}
[1,1,1]\left [ \begin{array}{c c c }
\frac{5}{6} & \frac{1}{6} & 0 \\
\frac{1}{3} & \frac{1}{3} & \frac{1}{3} \\
0 & \frac{1}{6} & \frac{5}{6} 
\end{array} \right]^t.
\end{eqnarray}
This expression converges rapidly as $t$ increases \cite{bartonetheridge11}. Hence, the vector resulting from the calculation in Expression (\ref{eq:Mlimit}) above is stable to additional right-multiplications by $M$ for sufficiently large $t$. This vector is the vector of reproductive values and when normalized, yields the probability distribution of the origin of a randomly-chosen individual. This is captured in the following definition, which is a common contemporary version of Fisher's original reproductive value \cite{rousset04, bartonetheridge11}.

\begin{defn}
\label{defn:RV}
Let $G$ be a graph and $M$ be the backward neutral transition probability matrix defined above. The reproductive value of individual $i$ is the $i$th entry $V_i$ of the non-zero solution vector $V$ of the equation $V = V M$. That is, $V$ is the left eigenvector of $M$. 
\end{defn}

It is worth noting that the equation $V=VM$ does not have a unique solution for $V$; any non-zero multiple $c$ of a solution $V_0$ is also a solution. Therefore, reproductive values are understood throughout this article as relative values. 

In the neutral process on a graph, some vertices may be favoured by the population dynamics and the individual residing on such a vertex can expect to have a greater number of offspring. These natural differences need to be accounted for in an evolutionary analysis. In a non-neutral case, where the evolutionary outcome is determined by differences in fitness, some vertices may bestow a natural fitness advantage to the resident irrespective of the resident's trait value. Thinking in terms of evolutionary game theory, individuals residing on vertices interact along edges and experience gains and losses to fitness due to these interactions. These gains and losses may differ depending on who is receiving the benefit/cost \cite{taylor90}. An individual on a high-degree vertex may experience a loss of fitness, but this may be offset by the natural fitness advantage of residing on a high-degree vertex. These environment-mediated fitness differences must first be understood before proceeding with non-neutral evolutionary processes. 

If the population structure is very symmetric---like the lattice structure in Figure \ref{fig:lattice}---then all individuals have identical reproductive output. This is not the case for general, non-symmetric graphs, such as the line $3$-line graph in Figure \ref{fig:3line} or the wheel graph in Figure \ref{fig:wheel}. In those examples, the differences in degrees results in differences in how often an individual replaces, or is replaced by an offspring of, another individual. These differences in fitness are accounted by reproductive value.

\begin{figure}[ht]
\centering

\includegraphics[width=0.3 \textwidth]{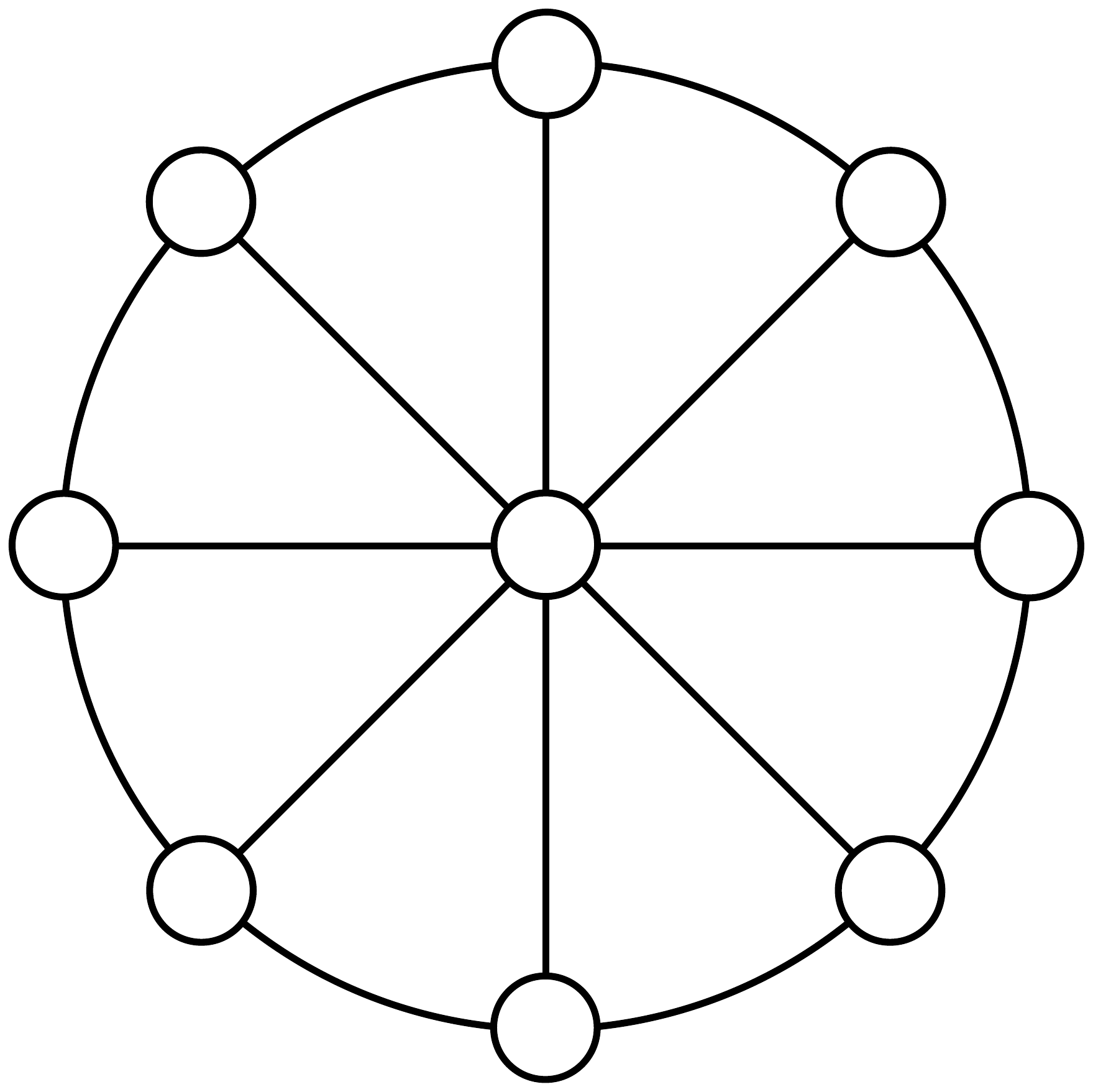} 
\caption{The wheel graph on $9$ vertices.}
\label{fig:wheel}
\end{figure}

As an illustrative example, consider the wheel graph in Figure \ref{fig:wheel}. There are two types of vertices, those on the periphery, denoted $v_P$, and the lone centre, denoted $v_H$. Consider, in turn, both a birth-death and a death-birth Moran process \cite{moran58} on this graph, and suppose that the population is neutral. In the birth-death process an individual is chosen at random to give birth and the resulting offspring displaces an adjacent neighbour at random. The individual at vertex $v_H$ is chosen to give birth with probability $1/9$, yet its neighbours are selected with probability $8/9$. Once a $v_P$ resident is selected, it displaces $v_H$ with probability $1/3$. In the death-birth process $v_H$ is chosen to die with probability $1/9$ but its neighbours are chosen with probability $8/9$. It would seem, then, that individuals at vertex $v_H$ is somehow ``better off'' in the death-birth than in the birth-death scheme. This is indeed the case. The way of quantifying ``better-off-ness'' is with reproductive value. 

\section{Metapopulations}

A metapopulation is a collection of \emph{demes} all linked by a dispersal pattern. Evolutionary graphs can often be thought of metapopulations where the vertices are demes and the edges are the dispersal pattern. Metapopulations were introduced by Levins \cite{levins70} as a means of describing populations with subpopulations experiencing extinction and re-colonization. Since \cite{levins70}, the scope and generality of metapopulation models has increased dramatically; see \cite{hanski98} for an introduction. 

Consider a metapopulation consisting of $N$ demes $v_1, v_2, \dots, v_N$. Each deme $v_i$ is a well-mixed population of fixed size $N_i$. The total population size is a constant, $N_{tot}$. After reproduction the offspring migrate to another deme with probability $m$ or stay on their natal deme with probability $1-m$ and for simplicity I assume the value of $m$ is identical for all demes. 


There are many possible population dynamics, for example, the Wright-Fisher process \cite{wright31}, imitation dynamics \cite{antalrednersood06}, and the Cannings process \cite{cannings74}. I restrict the focus of this article to two: the Moran death-birth and birth-death processes \cite{moran58}. In the death-birth process an individual is chosen at random to die. Suppose this individual resides on deme $v_i$. With probability $1-m$ the newly vacated site is occupied by the offspring of a deme mate. With probability $m$ it is occupied by the offspring of a member of a neighbouring deme $v_j$ chosen according to its relative size,
\begin{eqnarray}
\dfrac{N_j}{\displaystyle \sum_{k\in \mathcal{N}(v_i)}N_k},
\label{eq:disp}
\end{eqnarray}
where the sum is taken over all deme $v_i$'s neighbouring demes $\mathcal{N}(v_i)$. 

In the above definition I have assumed uniform dispersal probabilities to a deme. That is, if an individual on $v_i$ dies and is not replaced by a deme mate, then it is replaced by the offspring of a neighbouring deme $v_j$ with probability proportional to $v_j$'s size relative to the other neighbours of $v_i$. It is possible, however, that offspring are more likely to come from certain demes, regardless of the resident population size. 

Denote the probability that an individual chosen to die on $v_i$ is replaced by the offspring of an individual from $v_j$, conditional on the individual not being replaced by the offspring of another individual on deme $v_i$, by $w_{ji}$. If an individual dies on deme $v_i$ and the empty site is not taken by the offspring of a deme $v_i$ individual, then it is taken by the offspring from a neighbouring deme. Hence,
\begin{eqnarray}
\sum_{j\neq i} w_{ji} = 1.
\end{eqnarray}
With this notion of non-uniform dispersal probability, the probability that a newly vacated site on deme $i$ is occupied by an offspring of deme $j$ is given by
\begin{eqnarray}
\dfrac{w_{ji}N_j}{\displaystyle \sum_{k\in \mathcal{N}(v_i)}w_{ki}N_k}.
\label{eq:wdisp}
\end{eqnarray}

In the birth-death process an individual is chosen at random to reproduce and the new offspring either stays on its natal deme with probability $1-m$ and displaces a deme-mate or disperses to a neighbouring deme with probability $m$. The neighbouring deme is chosen according to the dispersal probabilities $u_{ij}$. Define $u_{ij}$ to be the probability that an offspring produced on $v_i$ disperses to and replaces an individual on deme $v_j$, conditional on the offspring not staying and replacing an individual on $v_i$. Note that, similar to the above,
\begin{eqnarray}
\sum_{j\neq i} u_{ij} = 1.
\end{eqnarray}

It is to be kept in mind that the $u_{ij}$ are the dispersal probabilities in the birth-death process while the $w_{ji}$ are in the death-birth process. In general, the $u_{ij}$ and $w_{ji}$ are not equal; a distinction I will draw in the next section. The difference between the two may seem subtle---$u_{ij}$ is the probability that an offspring produced on deme $i$ displaces an individual on deme $j$, while $w_{ji}$ is the probability that an empty site on deme $j$ is filled by an offspring from deme $i$---but must be kept in mind. The real difference between the $u_{ij}$ and $w_{ji}$ is the $u_{ij}$ are normalized with respect to the deme dispersed from, while $w_{ji}$ is normalized with respect to the deme dispersed to. This distinction allows Equations (\ref{eq:DB}) and (\ref{eq:BD}) to be easily generalized to graph-structured populations. Note that for uniform disperal probabilities on degree-regular graphs, $u_{ij} = w_{ji}$. 

I now derive equations for the reproductive values in the Moran death-birth and birth-death processes in metapopulations. To do this, I define a matrix $M$ similar to that in Definition \ref{defn:RV}, but where the entries are the indexed by demes, not individuals. Specifically, the $i$, $j$ entry of $M$ is the probability that a randomly chosen individual on deme $v_j$ was from the deme $v_i$ before a birth/death event. Definition \ref{defn:RV} then yields a reproductive value $V_i$ for each deme $v_i$. To find the reproductive value of an individual on deme $v_i$, simply divide the deme reproductive value by the size of the deme, $V_i/N_i$. In all, this yields the following. 
\begin{theorem}
\label{thm:RV}
Consider a metapopulation of size $N_{tot}$ residing on $N$ demes structured according to some graph $G$. Deme $v_i$ is of size $N_i$, where $1\leq i \leq N$. Denote the reproductive value of deme $v_i$ by $V_i$. 
\begin{enumerate}
\item For the death-birth process, the $V_i$ satisfy
\begin{eqnarray}
\label{eq:DB}
V_i = \sum_{j\in\mathcal{N}(v_i)}\dfrac{w_{ij}N_i}{\displaystyle \sum_{k\in\mathcal{N}(v_j)}w_{kj}N_k}V_j.
\end{eqnarray}
\item For the birth-death process, the $V_i$ satisfy
\begin{eqnarray}
\left ( \sum_{j\in\mathcal{N}(v_i)} u_{ji} N_j\right )V_i = N_i\sum_{j\in\mathcal{N}(v_i)}u_{ij} V_j.
\label{eq:BD}
\end{eqnarray}
\end{enumerate}
In both cases the sums are taken over all neighbours $\mathcal{N}(v_i)$ of $v_i$, or neighbours $\mathcal{N}(v_j)$ of $v_j$.
\end{theorem}

\begin{proof}
This is done by simply calculating the columns of the matrix $M$. I demonstrate this for the death-birth process only, since Equation (\ref{eq:BD}) is found in a similar way. Entry $j$ in the $i$th column of $M$ is the probability $p_{ji}$ that an individual currently in deme $v_j$ was in deme $v_i$ before the death-birth event. For the entry $p_{ii}$, an individual on deme $v_i$ either was unaffected by the death-birth event, with probability $(N_{tot}-1)/N_{tot}$, or is the offspring of a $v_i$ deme mate, with probability $(1-m)/N_{tot}$. For all $p_{ji}$ with $j\neq i$, an individual on deme $v_j$ is the offspring of a deme $v_i$ individual with probability 
\begin{eqnarray}
\frac{m}{N_{tot}}\dfrac{w_{ij}N_i}{\displaystyle \sum_{k\in\mathcal{N}(v_j)}w_{kj}N_k}.
\end{eqnarray}
Substituting these expressions into the backward transition probability matrix $M$ and evaluating the equation for reproductive value in Definition \ref{defn:RV} for $V_i$ yields,
\begin{eqnarray}
V_i =\dfrac{N_{tot}-1}{N_{tot}}V_i + \dfrac{1-m}{N_{tot}}V_i + \frac{m}{N_{tot}}\sum_{j\in\mathcal{N}(v_i)}\dfrac{w_{ij}N_i}{\displaystyle \sum_{k\in\mathcal{N}(v_j)}w_{kj}N_k}V_j.
\end{eqnarray}
Simplifying gives Equation (\ref{eq:DB}) in Theorem 1.
\end{proof}

Theorem 1 demonstrates how the reproductive values in a metapopulation depend only on the size of the demes and the rates of dispersal. An interesting example to consider is a heterogeneous metapopulation that has all demes of the same reproductive value. Suppose such a metapopulation is structured according to the wheel graph of the previous section. Deme $v_H$ is of size $N_H$ and $v_P$ is of size $N_P$. Setting $V_P = V_H$ in the equation in Theorem 1 that describes the death-birth process yields a system of equations for $N_P$ and $N_H$ with solution $N_H = 6N_P$. That is, in a metapopulation structured according to the $9$-wheel graph, the reproductive values of all the demes are equal provided $N_H = 6N_P$. The individual reproductive values are obtained by dividing the deme RVs by the size of the deme. In this way it is seen that an individual in a periphery deme in a population undergoing a death-birth process has a greater reproductive value than on in the hub, despite both being members of deme with the same average reproductive value.

\section{Graph-structured Populations}

A graph-structured population is a special case of a metapopulation with $N_i=1$ for all $i$ and $N_{tot}= N$. There are a couple of ways we can analyse the reproductive value equations in Theorem 1 in the context of evolutionary graphs. First I consider the case that the probability of offspring dispersal from a vertex to a neighbouring vertex is uniform. That is, I set 
$$
w_{ji} = u_{ij} =  \left \{\begin{array}{c c}
         1/d_i & \hbox{if } v_i \hbox{ and } v_j \hbox{ are adjacent} \\
         0 & \hbox{otherwise}
	  \end{array} \right. ,
$$
where $d_i$ is the degree of vertex $i$ and $w_{ji}$ and $u_{ij}$ are the death-birth and birth-death dispersal probabilities, respectively, defined in the previous section. This yields the following solutions to the equations in Theorem 1. 

\begin{cor}
For an evolutionary graph with uniform dispersal from any vertex the reproductive values $V_i$ for the vertices $v_i$ of degrees $d_i$ are as follows.
\begin{enumerate}
\item For the death-birth process, 
\begin{eqnarray}
V_i = d_i
\end{eqnarray}
\item For the birth-death process,
\begin{eqnarray}
V_i = \frac{1}{d_i}.
\end{eqnarray}
\end{enumerate}
\end{cor}
This corollary is very useful in describing the neutral process, which will be done next. First note that the equations in Theorem 1 have a degree of freedom, so there are an infinite number of solutions. But they are all scalar multiples of those given above.

I now consider the relationship between reproductive value and fixation probability. Suppose a population consists entirely of one type (type $B$) of individual. After a reproductive event a mutant (type $A$) appears. The probability that the progeny of the mutant go on to displace all resident types is the fixation probability $\rho_A$ of $A$. In general this fixation probability depends on where in the population the $A$ type emerges. Define $\rho_{A|i}$ as the fixation probability of an $A$ that emerges on vertex $v_i$. 

It is known (ex. \cite{leturquerousset02}) that the fixation probability of a neutral mutant in a metapopulation is equal to its relative reproductive value. This fact can easily be seen to be the case from a result of \cite{broomhadjichrysanthourychtarstadler10}. 

\begin{theorem}
Let $G$ be an evolutionary graph with $N$ vertices and suppose the edges are uniformly weighted. The fixation probability $\rho_{A|i}$ of a single $A$ type that emerges on vertex $v_i$ of $G$ in the neutral population is
\begin{eqnarray}
\label{eq:conditionalfixation}
\rho_{A|i} = \dfrac{V_i}{\sum_{j=1}^N V_j},
\end{eqnarray}
where $V_i$ is the reproductive value of vertex $v_i$. 
\label{thm:neutral}
\end{theorem}
A proof of this result is in the appendix. 

The fundamental question in evolutionary theory is, when does a mutant have an evolutionary advantage over a resident population? A natural condition is that the probability $\rho$ that the mutant fixes in the popultaion is greater than what it would be in the absence of selection. From early on in the evolutionary graph theory literature \cite{liebermanhauertnowak05} this condition took the form $\rho_A > 1/N$, where $A$ is the mutant and $N$ is the total population size. Theorem 2 indicates that this condition is insufficient for graphs with vertices of differing degrees. For an arbitrary graph, $\rho_{A|i}$ depends on $i$.  Notice, however,
\begin{eqnarray}
 \frac{1}{N}\sum_{i=1}^N\rho_{A|i} = \frac{1}{N}\sum_{i=1}^N\dfrac{V_i}{ \sum_{j=1}^NV_j}= \frac{1}{N}.
\end{eqnarray}

Returning to the wheel graph example, Theorem 2 allows for an easy calculation of the neutral fixation probability of a hub $\rho_{A|H}$ or periphery $\rho_{A|P}$ mutant on a wheel graph of arbitrary size, $N+1$. Table \ref{table} records these fixation probabilities for both the birth-death and death-birth processes.

\begin{table}[h]
\centering
\begin{tabular}{c | c | c |}
& DB & BD \\ 
\hline
$\rho_{A|H}$ & $1/4$ & $3/(N^2+3)$ \\ 
\hline
$\rho_{A|P}$ & $3/(4N)$ & $N/(N^2+3)$ \\
\hline
\end{tabular}
\caption{The fixaion probabilities for an allele that begins on a hub or on a periphery vertex for both the birth-death and death-birth Moran processes.}
\label{table}
\end{table}

A few interesting observations can be made at this point. First, in the death-birth process $\rho_{A|H}$ does not depend on the size of the population. This is understood as a balance between the probabilities that the hub or a periphery individual is chosen to die. For large populations the probability that the hub dies is essentially zero, yet the probability that the hub reproduces is fixed at $1/3$.  Second, for the birth-death process, both $\rho_{A|P}$ and $\rho_{A|H}$ go to $0$ as $N$ increases. This is because for large populations the probability that any one individual is chosen to reproduce in close to $0$.

An interesting extension of Theorem \ref{thm:neutral} is to the neutral fixation probability of a set $M$ of $A$ types. Such a fixation probability is defined as the probability that the population eventually consists entirely of all $A$ given that it initially started with a set $M\subset V(G)$ of $A$s. 
\begin{theorem}
The neutral fixation probability $\rho_{A|m}$ of a set $M$ of $A$ types on a graph $G$ undergoing either a birth-death or death-birth Moran process is 
\begin{eqnarray}
\rho_{A|M} = \sum_{i\in M} \rho_{A|i}.
\label{eq:multineutral}
\end{eqnarray}
\label{thm:multineutral}
That is, the neutral fixation probability of a set of $A$ types is the sum of the individual neutral fixation probabilities. 
\end{theorem}
A proof of this theorem in found in the appendix.

\begin{figure}
\centering
\subfigure[]{\includegraphics[width=0.3\textwidth]{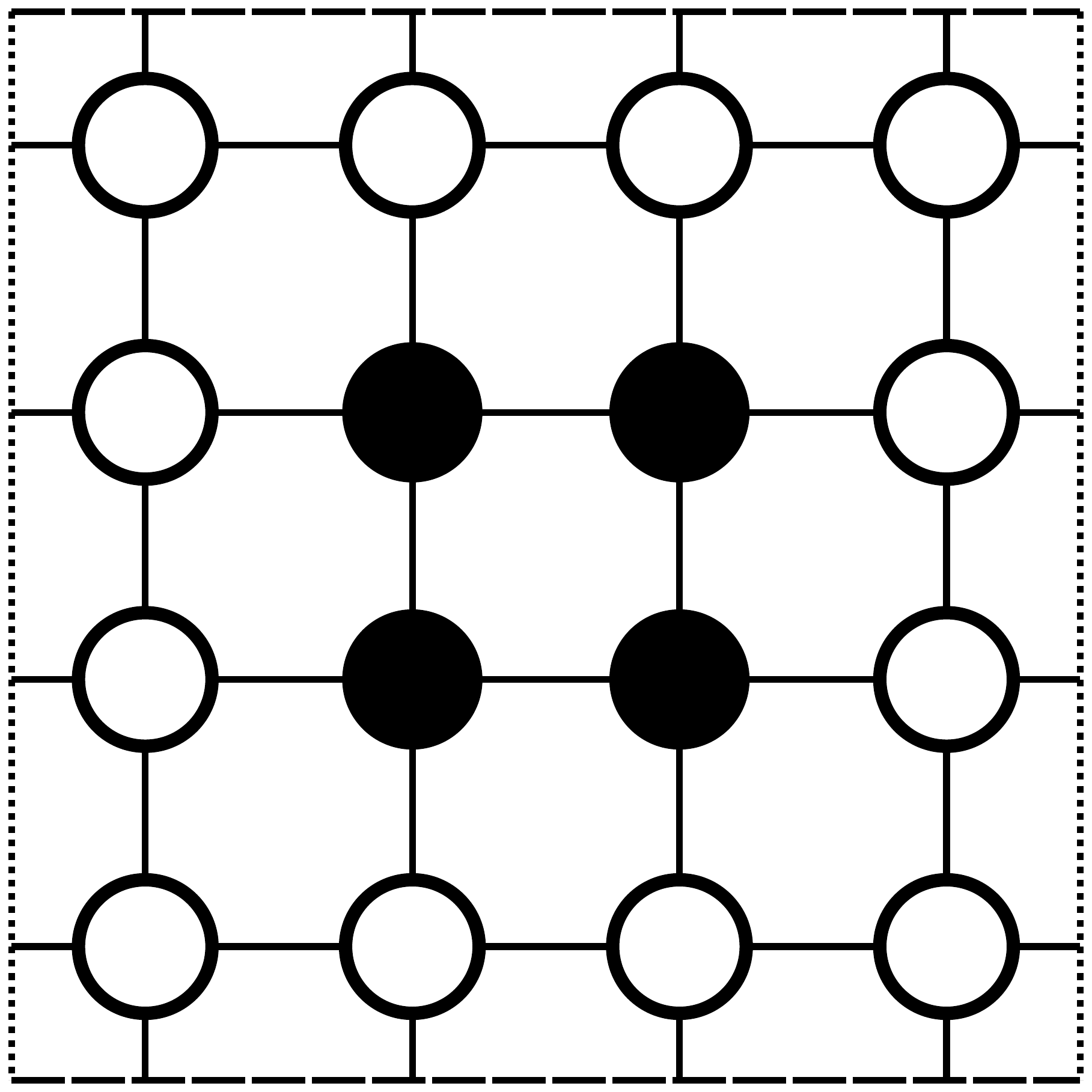}}
\hspace{1cm}
\subfigure[]{\includegraphics[width=0.3\textwidth]{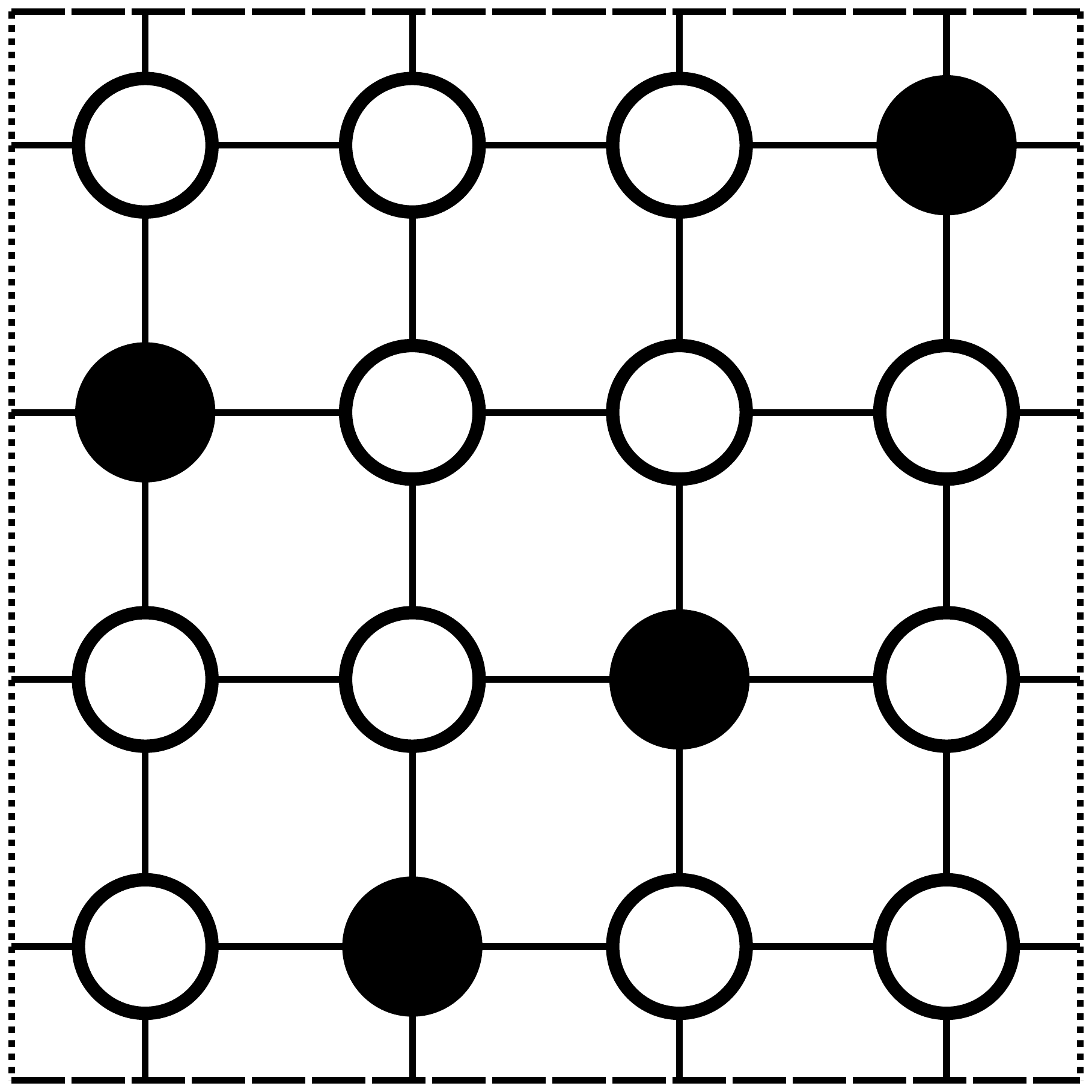}}
\caption{The fixation probability of a set of individuals is the sum of the fixation probabilities of the individuals in the set. In this example on a lattice, the set of black individuals has the same fixation probability whether they are clumped (a) or spread out (b). }
\label{fig:lattice}
\end{figure}

This theorem is remarkable in that the configuration of the $A$ types is irrelevant. It does not matter if the set $M$ is clustered or spread about the graph; the fixation probability is the same; see Figure \ref{fig:lattice}.

\section{Non-neutral Cases}

So far I have analysed the equation in Theorem 1 by supposing that the dispersal from any vertex $i$ to a neighbouring $j$ is $w_{ji}=u_{ij}=1/d_i$. This need not be the case. One could imagine a population residing in a windy or a stream environment that results in preferential dispersal. Removing the assumption of uniform dispersal makes Theorem 1 less transparent. Relating reproductive value to an existing object in the study of evolutionary graphs, the \emph{temperature} of vertices as introduced in \cite{liebermanhauertnowak05}, allows us to gain some traction.

For the birth-death process, the temperature $T_i$ of a vertex $v_i$ is 
\begin{eqnarray}
T_i = \sum_{j\in \mathcal{N}(v_i)} u_{ji},
\label{eq:utemp}
\end{eqnarray}
where the sum is over all neighbours of $v_i$. If the graph is weighted with $w_{ji}$ weights, as in the death-birth process, the above definition can be rewritten accordingly:
\begin{eqnarray}
T_i = \sum_{j\in \mathcal{N}(v_i)} w_{ij}.
\label{eq:wtemp}
\end{eqnarray}
Equation (\ref{eq:wtemp}) is not the definition of temperature as found in, e.g. \cite{liebermanhauertnowak05, nowak06}. Previous work on the temperature of vertices has only considered the birth-death process. As we have seen in Equation (\ref{eq:DB}) it is necessary to introduce $w_{ji}$ for the death-birth process. Recall that $\sum_{j\in \mathcal{N}(v_i)} w_{ji} = 1$. The $T_i$ in Equation (\ref{eq:wtemp}) does not necessarily equal $1$, and therefore plays the same role for the $w_{ij}$ that the temperature $T_i$ in Equation (\ref{eq:utemp}) does for the $u_{ji}$. It can be shown that the existing results on temperature, including Theorem \ref{thm:isothermal} below, also hold for graphs carrying the $w_{ji}$ weightings.

The fundamental result concerning the temperatures on an evolutionary graph is the isothermal theorem of \cite{liebermanhauertnowak05} (see also \cite{nowak06}). Suppose a mutant with fecundity $r$, where $r>1$, emerges in a population of individuals each having fecundity $1$. The population updates with a Moran process and the probability $\rho$ that the mutant fixes in the population is observed.  This is the constant-fecundity process \cite{liebermanhauertnowak05}.

The results of \cite{liebermanhauertnowak05} are that for an \emph{isothermal} graph, where all vertices have the same temperature, the fixation probability is exactly what one would find in a unstructured population---that is, where all verticies are adjacent; a \emph{complete} graph---of the same size, $N$.
\begin{theorem} (Lieberman et al., 2005)
Let $G$ be a graph and $T_i$ be the temperature of the vertex $v_i$. For the constant-fecundity process described above,
\begin{eqnarray}
\rho_A = \dfrac{1-1/r}{1 -1/r^N}
\end{eqnarray}
if, and only if, 
\begin{eqnarray}
T_i = T_j \ \ \forall \ i, j \in V(G).
\label{eq:isothermal}
\end{eqnarray}
\label{thm:isothermal}
\end{theorem}
Equation (\ref{eq:isothermal}) is the isothermal condition. This relates nicely to reproductive value.

\begin{theorem}
A graph is isothermal if, and only if, all vertices have the same reproductive value. 
\end{theorem}

\begin{proof}
First, assume $V_i = V_j$ for all vertices $v_i$ and $v_j$ of $G$. From Equation \ref{eq:DB}, I have
\begin{eqnarray}
V_i = \sum_{j\in \mathcal{N}(v_i)}w_{ij}V_j \ \Longrightarrow \ \underbrace{\sum_{j\in\mathcal{N}(v_i)}w_{ij}}_{T_i} = 1
\end{eqnarray}
for the death-birth process, and
\begin{eqnarray}
\left ( \sum_{j\in\mathcal{N}(v_i)} u_{ji} \right )V_i = \sum_{j\in\mathcal{N}(v_i)}u_{ij} V_j \ \Longrightarrow \ \underbrace{\sum_{j\in\mathcal{N}(v_i)}u_{ji}}_{T_i} = \sum_{j\in \mathcal{N}(v_i)} u_{ij} = 1
\end{eqnarray}
for the birth-death process. 

Now assume $T_i = T_j$ for all vertices $v_i$ and $v_j$ of $G$. Suppose, for contradiction, that not all vertices of $G$ have the same reproductive value. There exists a vertex $v_k$ such that $V_k$ is the maximum of all reproductive values and at least one neighbour of $v_k$ has a reproductive value strictly less than $V_k$. Similarly, let vertex $v_l$ be such that $V_l$ is the minimum of all reproductive values and at least one neighbour of $v_l$ has a reproductive value greater than $V_l$. Consider the death-birth process; the argument for the birth-death process is similar. From Equation \ref{eq:DB}, 
\begin{eqnarray}
V_k = \sum_{j\in\mathcal{N}(v_k)} w_{kj} V_j <  \sum_{j\in\mathcal{N}(v_k)} w_{kj} V_k \ \Longrightarrow  \ T_k= \sum_{j\in\mathcal{N}(v_k)} w_{kj} > 1.
\end{eqnarray}
Also,
\begin{eqnarray}
V_l = \sum_{j\in\mathcal{N}(v_l)} w_{lj} V_j >  \sum_{j\in\mathcal{N}(v_l)} w_{lj} V_l \ \Longrightarrow \  T_l=\sum_{j\in\mathcal{N}(v_l)} w_{lj} < 1.
\end{eqnarray}
Hence, $T_k \neq T_l$, which is a contradiction. 
\end{proof}

In \cite{broomrychtar08} the authors prove that, assuming dispersal from a vertex is uniform, a graph is isotermal if, and only if, the graph is regular. In light of Corollary 1 or Theorem 4, an analogous result exists for reproductive value. An interesting question is, is it possible for the vertices of a non-regular graph to all have the same reproductive value? The answer is yes, as is seen by, once again, returning to the wheel graph example. For the wheel graph on $9$ vertices in Figure 1 consider the birth-death process and define the dispersal probabilities
$$
u_{PH} = \frac{1}{8}, \ u_{PP} = \frac{7}{16}, \ \hbox{ and,} \ u_{HP} = \frac{1}{8}.
$$
This example is easily seen to be isothermal and hence, by Theorem 4, all vertices have the same reproductive value. 

For the death-birth process and constant fecundity, higher-degree vertices are favoured for the emergence of more fecund alleles, since they confer a natural advantage: higher degree means a greater likelihood of a neighbour dying which translates into a greater-than-average chance of placing an offspring. The situation is reversed for the birth-death process: lower degree means less-than-average chance of being displaced by a neighbour's offspring. In both cases the favourable vertex is one with a high reproductive value. 

Previous work on degree-heterogeneous graphs \cite{broomrychtarstadler11, antalrednersood06} has reached the conclusion that the death-birth process favours mutants that emerge on vertices of high degree, while the birth-death process favours mutants arising on low degree vertices. This is precisely what is found in the corollary to Theorem \ref{thm:RV}. However, rather than viewing the results of \cite{broomrychtarstadler11, antalrednersood06} as two separate cases, the above results on reproductive value allow us to observe a general phenomenon: the vertices that favour the mutant allele in the constant-selection framework are those with the greatest reproductive value, regardless if the update rule is death-birth or birth-death. The difference in fecundity acts to embelish the effect of reproductive value. Reproductive value provides a unifying concept for these results.

\subsection{Evolutionary Games}

I now consider evolutionary games on graphs. Consider a population consisting of two types of individuals: $A$s and $B$s. Each pair connected by an edge gives and receives payoffs according to the game matrix
\begin{eqnarray}
\label{eq:prisonersdilema}
\begin{array}{c | c | c |}
\ & A & B \\
\hline
A & b-c & -c \\
\hline
B & b & 0 \\
\hline
\end{array}
\end{eqnarray}
For $b,c>0$, this is the additive prisoner's dilemma game. 

The payoffs accrued by individuals interacting according to the game in Matrix (\ref{eq:prisonersdilema}) translate into fecundity. The fecundity of an individual $i$ is 
\begin{eqnarray}
f_i=e^{\delta P},
\label{eq:fecundity}
\end{eqnarray} 
where $\delta$ is the strength of selection and $P$ is the payoff received from playing the game with their neighbours \cite{liebermanhauertnowak05}. For example, if an $A$ individual has one $A$ and two $B$ neighbours then their total payoff is $P= b-c + 2(-c)=b-3c$. Once these fecundity values are calculated, a population update occurs. For the death-birth process, an individual $i$ dies with probability $1/N$ and is replaced by an offspring of its neighbour $j$ with probability
\begin{eqnarray}
\frac{f_j}{f_{tot}},
\end{eqnarray}
where $f_{tot}$ is the total fecundity of all the neighbours of $i$. For the birth-death process, an individual $i$ is chosen for reproduction with probability 
\begin{eqnarray}
\frac{f_i}{f_{tot}},
\end{eqnarray}
and the offspring displaces a neighbour of $i$ with uniform probability, $1/d_i$.

To illustrate the effects of reproductive value on the outcome of an evolutionary game, I consider the simplest example of a heterogeneous graph, the $3$-line in Figure {\ref{fig:3line}}. Denote a end point vertex with the subscript $p$ and the central hub vertex with $h$. I consider only the birth-death process.

The reproductive values for the birth-death process are easily calculated from Corollary 1:
\begin{eqnarray}
V_p = 1, \ \hbox{ and, } \ V_h = \frac{1}{2}.
\end{eqnarray}
Hence, in the neutral process, where $\delta=0$ in Equation (\ref{eq:fecundity}),
\begin{eqnarray}
\rho_p= \frac{2}{5}, \ \hbox{ and, }  \ \rho_h = \frac{1}{5}.
\end{eqnarray}
The neutral fixation probabilies $\rho_{neutral}$ give us a condition for the spread of the strategy $A$: $A$ is favoured by evolution provided $\rho_A> \rho_{neutral}$. Note that, for regular graphs, a class of graph that includes complete, cycles, lattices, and all vertex-transitive graphs, this condition reduces to $\rho_A > 1/N$, which is the condition commonly found in the literature \cite{nowak06}. In the present example, an $A$ type is favoured by evolution provided $\rho_{A|p}>\rho_{neutral|p}=2/5$ if it emerges on an end vertex and $\rho_{A|h}>\rho_{neutral|h}=1/5$ on the hub vertex.

I now calculate the probability that a single $A$ reaches fixation in a population otherwise comprised of all $B$. To do this, I assume weak selection. This means that $\delta  \ll 1$ in Equation (\ref{eq:fecundity}). This allows for an accurate Taylor series approximation for Equation (\ref{eq:fecundity}). 

Now suppose that the $3$-line population initially consits of all $B$. A mutant $A$ appears on one of the end point vertices. It is easy to show, by solving a system of equations that describes the fixation probability, that the fixation probability of this mutant is
\begin{eqnarray}
\rho_{A|p} = \frac{2}{5} - \left(\frac{14}{25}c+\frac{6}{25}b\right) \delta + o(\delta).
\end{eqnarray}
The condition $\rho_{A|p}>\rho_{neutral|p}$ yields a condition on the $b$ and $c$ parameters. Namely, an $A$ type on an end point vertex is favoured by evolution provded $b/c<-7/3$. This condition is clearly never satisfied for positive $b$ and $c$. If, however, $b<0$, then the condition can be satisfied. This is an example of a \emph{spiteful} behaviour: an individual pays a cost to purposely harm another \cite{gardnerwest06}. 

A similar calculation reveals that the fixation probability of an $A$ type that emerges on the hub vertex is 
\begin{eqnarray}
\rho_{A|h} = \frac{1}{5} - \left(\frac{12}{25}c+\frac{14}{75}b\right) \delta + o(\delta).
\end{eqnarray}
This $A$ is favoured by evolution provided $b/c<-18/7$. Again, this is satisfied only when $b<0$. 

To compare these two results for the fixation probability of spite, suppose that the cost of the spiteful act is fixed at $c=1$. Then it is seen that the hub requires a higher level of spite than the end point vertices in order for the trait to fix in the population. Put another way, spite can emerge more easily on the end point vertices. 

The lesson from this example is that the spread of strategy may be tied to where the strategy first emerges. In turn, the favoured vertices are those with a greater reproductive value.

\section{Discussion}

In this article, I have brought the notion of reproductive value into the study of evolution in graph-structured populations. This makes headway into unifying existing results on degree-heterogeneous graphs. Generally, for a constant-fecundity process in a graph-structured population, it may be best for the more fecund type to emerge on a vertex with high reproductive value. This depends on the type of population regulation. For birth-death updating, mutants are favoured on low-degree vertices, while mutants in the death-birth process are favoured high-degree vertices. This has been observed by other authors \cite{antalrednersood06}, but as separate cases. Reproductive value unites these into two sides of the same coin. 

The main driving force of these differences is the neutral fixation probability. Some breeding sites are more advantageous to occupy in that they naturally confer a fitness advantage on their resident. This natural advantage is captured by the reproductive value of an individual on such a site. 

The effect of heterogeneous population structures is still not well understood. It is now well-known that degree-heterogeneous graphs can affect evolution \cite{antalrednersood06, santossantospacheco08, broomrychtar08, broomrychtarstadler11}, but an explanation of how the degrees of individual vertices contribute to these effects is still needed. The concept of reproductive value fills this void. In the neutral process, those individuals that reside on vertices of a higher reproductive value have a higher-than-average probability of fixing in the population. Understanding the neutral process allows for a baseline condition against which the fitness advantage of a non-neutral allele can be measured. Where the allele emerges matters. If the allele has a fixation probability greater than the relative reproductive value of the vertex on which it emerges, it is favoured by evolution. Such a condition opens up the further study of evolution in heterogeneous graph-structured populations.    

This work also clarifies terminology existing in the literature. Take, for example, a statement from \cite{santossantospacheco08}, ``For regular graphs (in which, from the perspective of a population structure, every individual is equivalent to any other)...'' The meaning of such a statement is not entirely clear. The statement is true for \emph{vertex-transitive} graphs, as shown in \cite{taylordaywild07}, where it is understood that the graph ``looks the same'' from every vertex. Reproductive value formalizes the idea present in the above statement: on a regular graph, all individuals have the same reproductive value. It should be noted that all vertex-transitive graphs are regular, but not all regular graphs are vertex-transitive; an example is the Frucht graph \cite{frucht39}.  There are a host of factors that influence the outcome of an evolutionary process on a graph: the graph structure, including symmetry, the degree distribution, and the underlying structure; the population regulation scheme; and whether the population is experiencing constant or frequency-dependent selection. All of these factors need to be stated carefully to avoid the misinterpretation of results.

\section{Acknowledgements}
I am indebted to Christoph Hauert, Peter Taylor, Lucas Wardil, and to an anonymous reviewer for supplying valuable comments on drafts of this article. This work is supported by the National Sciences and Engineering Research Council of Canada.

\bibliographystyle{plain}
\bibliography{thesis}

\section{Appendix}

\subsection{Proof of Theorem \ref{thm:neutral} and \ref{thm:multineutral}}
I first prove Theorem \ref{thm:multineutral} and then use this result in the proof of Theorem \ref{thm:neutral}.

\begin{thm3}
The neutral fixation probability $\rho_{A|m}$ of a set $M$ of $A$ types on a graph $G$ undergoing either a birth-death or death-birth Moran process is 
\begin{eqnarray}
\rho_{A|M} = \sum_{i\in M} \rho_{A|i}.
\label{eq:multineutral2}
\end{eqnarray}
That is, the neutral fixation probability of a set of $A$ types is the sum of the individual neutral fixation probabilities. 
\end{thm3}

In preparation for this proof, define the \emph{state} of the population to be the set of $A$ types in the population. For all states $\mathcal{S}$, the fixation probability $\rho_{\mathcal{S}}$ of the set $\mathcal{S}$ satisfies the equation
\begin{eqnarray}
\rho_{\mathcal{S}}= \sum_{\mathcal{T}\neq \mathcal{S}} P_{\mathcal{S},\mathcal{T}} \ \rho_{\mathcal{T}} +\left(1-\sum_{\mathcal{T}\neq \mathcal{S}} P_{\mathcal{S},\mathcal{T}} \right)\rho_{\mathcal{S}}.
\label{eq:stateeq}
\end{eqnarray}
As an explicit instance of this equation, consider a well-mixed population of size $N$. The states are precisely the number of $A$ types. Equation (\ref{eq:stateeq}) is then
\begin{eqnarray}
\rho_i = P_{i,i+1}\rho_{i+1} +P_{i,i-1}\rho_{i-1} +(1-P_{i,i+1}-P_{i,i-1})\rho_i,
\end{eqnarray}
which is found elsewhere in the literature \cite{nowak06}.

\begin{proof}
Considering all states of the population, Equation (\ref{eq:stateeq}) is a system of equations. For the initial conditions $\rho_{0} = 0$ and $\rho_N = 1$, where $\rho_0$ is the state with no $A$ types and $\rho_N$ is the state of all $A$ types, then the system defined by Equation (\ref{eq:stateeq}) has a unique solution up to a non-zero constant. Hence, it suffices to show that Equation (\ref{eq:multineutral2}) satisfies these two initial conditions and the system defined by Equation (\ref{eq:stateeq}). 

Clearly, Equation (\ref{eq:multineutral2}) satisfies the two initial conditions. To show that it satisfies the system above, rewrite Equation (\ref{eq:stateeq}) as
\begin{eqnarray}
\sum_{\mathcal{T}\neq \mathcal{S}} P_{\mathcal{S},\mathcal{T}} \left(\rho_{\mathcal{S}} - \rho_{\mathcal{T}}\right) = 0.
\label{eq:stateeq2}
\end{eqnarray}
Note that the states $\mathcal{S}$ and $\mathcal{T}$ can differ by at most one vertex. For all other states, $\mathcal{T}'$, $P_{\mathcal{S},\mathcal{T}'}=0$. 

Denote the state obtained from state $\mathcal{S}$ by switching the type of individual that occupies vertex $j$ by $\mathcal{S}(j)$. With this, Equation (\ref{eq:stateeq2}) is
\begin{eqnarray}
\sum_{j} P_{\mathcal{S},\mathcal{S}(j)} \left(\rho_{\mathcal{S}} - \rho_{\mathcal{S}(j)}\right) = 0.
\label{eq:stateeq3}
\end{eqnarray}
I now substitute Equation (\ref{eq:multineutral2}) into the left-hand side of the above:
\begin{eqnarray}
\sum_{j} P_{M,M(j)} \left(\rho_{A|M} - \rho_{A|M(j)}\right)=0.
\label{eq:stateeq3}
\end{eqnarray}
Either $j\in M$ or $j\notin M$. In the first case,  
\begin{eqnarray}
\left(\rho_{A|M} - \rho_{A|M(j)}\right) = \rho_{A|M} - \rho_{A|M\setminus \{j\}},
\end{eqnarray}
and in the second,
\begin{eqnarray}
\left(\rho_{A|M} - \rho_{A|M(j)}\right) = \rho_{A|M} - \rho_{A|M\cup \{j\}}.
\end{eqnarray}
At this point, I require an expression for $P_{M,M(j)}$. This will depend on whether a birth-death or a death-birth process is being considered. For the birth-death process, 
\begin{eqnarray}
P_{M,M(j)} = \dfrac{1}{N}\displaystyle \sum_{k\in \mathcal{N}'(j)}u_{kj},
\label{eq:bdprob}
\end{eqnarray}
where the sum is taken over all neighbours of $j$ that are a different type than $j$. Substituting this into Equation (\ref{eq:stateeq3}) yields
\begin{eqnarray}
\dfrac{1}{N}\sum_{j} \displaystyle \sum_{k\in \mathcal{N}'(j)}u_{kj}\left(\rho_{A|M} - \rho_{A|M(j)}\right)
\end{eqnarray}
\begin{eqnarray} 
=\dfrac{1}{N}\left(\sum_{j\in M} \displaystyle \sum_{k\notin M}u_{kj}\left( \rho_{A|M} - \rho_{A|M\setminus \{j\}}\right )+ \sum_{j\in M} \displaystyle \sum_{k\notin M}u_{kj}\left (\rho_{A|M} - \rho_{A|M\cup \{j\}}\right )\right)=0.
\label{eq:intermediate}
\end{eqnarray}
At this point, I directly substitute Equation (\ref{eq:multineutral2}) into Equation (\ref{eq:intermediate}). With some simplification, Equation (\ref{eq:intermediate}) is
\begin{eqnarray} 
\dfrac{1}{N}\left(\sum_{j\in M} \displaystyle \sum_{k\notin M}u_{kj}\rho_{A|j} + \sum_{j\in M} \displaystyle \sum_{k\notin M}u_{jk}\left (- \rho_{A|j}\right )\right)=\dfrac{1}{N}\left(\sum_{j\in M} \displaystyle \sum_{k\notin M}\left(u_{kj}-u_{kj}\right ) \rho_{A|j} \right) = 0.
\end{eqnarray}
Hence, Equation (\ref{eq:multineutral2}) is a solution to Equation (\ref{eq:stateeq}) and is, therefore, the desired probability. 

\end{proof}

The argument above can be descibed as follows. Every instance of a vertex $j$ of $M$ being replaced by an individual $k$ not in $M$ exactly cancels with an instance of $j$ replacing $k$ to create the set $M\cup k$.

The argmuent for the death-birth process is analogous. The only difference is that Equation (\ref{eq:bdprob}) is
\begin{eqnarray}
P_{M,M(j)} = \dfrac{1}{N}\displaystyle \sum_{k\in \mathcal{N}'(j)}w_{kj}.
\label{eq:dbprob}
\end{eqnarray}

\begin{thm2}
Let $G$ be an evolutionary graph with $N$ vertices and suppose the edges are uniformly weighted. The fixation probability $\rho_{A|i}$ of a single $A$ type that emerges on vertex $v_i$ of $G$ in the neutral population is
\begin{eqnarray}
\rho_{A|i} = \dfrac{V_i}{\sum_{j=1}^N V_j},
\end{eqnarray}
where $V_i$ is the reproductive value of vertex $v_i$. 
\end{thm2}

\begin{proof}

I consider a death-birth process; the result for the birth-death process is derived analogously. A general proof that holds irrespective of the update rule can be derived from the results of \cite{leturquerousset02}. The proof of this theorem hinges on the Equation (\ref{eq:stateeq}). The following argument follows \cite{broomhadjichrysanthourychtarstadler10} where the authors prove a similar result for a birth-death process.

Similar to the previous proof, $\rho_{A|i}$ satisfies
\begin{eqnarray}
\rho_{A|i} =\frac{1}{N} \sum_{j\neq i} w_{ij}\rho_{A|\{ij\}} +\left( 1 - \frac{1}{N}\sum_{j\neq i} w_{ji}-  \frac{1}{N}\sum_{j\neq i} w_{ij} \right) \rho_{A|i}.
\end{eqnarray}
Rearranging yields
\begin{eqnarray}
\sum_{j\neq i} w_{ji}\rho_{A|i} = \sum_{j\neq i} w_{ij} \left( \rho_{A|\{ij\}} - \rho_{A|i} \right).
\label{eq:proofeq}
\end{eqnarray}
From the Theorem \ref{thm:multineutral}, I have $\rho_{A|\{ij\}} = \rho_{A|i} +\rho_{A|j}$. Combining this with the fact that $w_{ji}=w_{ij}=0$ for all non-adjacent $i$ and $j$, Equation  (\ref{eq:proofeq}) is
\begin{eqnarray}
\sum_{j\in \mathcal{N}(i)}w_{ji}\rho_{A|i}= \sum_{j\in \mathcal{N}(i)} w_{ij}\rho_{A|j}.
\end{eqnarray}
The solution for this is $\rho_{A|i}=d_i$. Normalizing by the sum of the degrees gives the result.

\end{proof}

\end{document}